\documentclass[journal]{IEEEtran}
\usepackage{amsmath,amsfonts,amssymb,mathrsfs,graphicx,cite}
\usepackage{multirow,algorithm,algpseudocode,url,subfigure, stfloats}
\usepackage{mathtools}
\usepackage{color}

\makeatletter
\let\MYcaption\@makecaption
\makeatother

\usepackage[pdfa]{hyperref}

\newcommand{\tr}{\text{\textnormal{tr}}}
\newcommand{\vect}{\text{\textnormal{vec}}}
\newcommand{\diag}{\text{\textnormal{diag}}}

\newcommand{\etal}{\textit{et.al. }}

\newtheorem{definition}{Definition}
\newtheorem{proposition}{Proposition}
\newtheorem{remark}{Remark}

\newtheorem{lemma}{Lemma}
\newtheorem{theorem}{Theorem}
\newtheorem{corollary}{Corollary}

\begin{document}
\title{Mean Square Stability for Stochastic Jump Linear Systems via Optimal Transport}
\author{Kooktae Lee, Abhishek Halder, and Raktim Bhattacharya% <-this % stops a space
\thanks{The authors are with the Department of Aerospace Engineering, Texas A\&M University, College Station, TX 77843-3141
USA. Their respective email addresses are animodor@tamu.edu, ahalder@tamu.edu, and raktim@tamu.edu. This research was supported through National Science Foundation
award \#1016299, with Dr. Helen Gill as the program manager.}}% <-this % stops a space
%\thanks{Manuscript received April 19, 2005; revised January 11, 2007.}}
\maketitle

\begin{abstract}
In this note, we provide a unified framework for the mean square stability of stochastic jump linear systems via optimal transport. The Wasserstein metric known as an optimal transport, that assesses the distance between probability density functions enables the stability analysis. Without any assumption on the underlying jump process, this Wasserstein distance guarantees the mean square stability for general stochastic jump linear systems, not necessarily for Markovian jump. The validity of the proposed methods are proved by recovering already-known stability conditions under this framework.
\end{abstract}


% \fi
\IEEEpeerreviewmaketitle

\section{Introduction}
Consider a discrete-time jump linear system given below.
\begin{equation}
x(k+1) = A_{\sigma_k}x(k)\label{eqn:1},\\
\end{equation}
where $\{\sigma_k\}$ denotes the switching sequence of the jump system.  This system (\ref{eqn:1}) with $m$ modes is characterized by (i) a set of $m$ system matrices $\{A_{i}\}_{i=1}^{m}$, and (ii) a switching sequence $\{\sigma_k\}$. If the switching sequence $\{\sigma_k\}$ is governed by the time-varying occupation probability vector $\pi\left(k\right) \triangleq \{\pi_{1}\left(k\right), \pi_{2}\left(k\right), \hdots, \pi_{m}\left(k\right)\}$, then \eqref{eqn:1} is referred as stochastic jump linear systems. In general, there are several difficulties on defining the conditions for the stability of such stochastic jump linear systems and main difficulty stems from the randomness in switching logic. Therefore, a variety of researches have been investigated for the stability analysis of stochastic jump linear systems.

Kozin \cite{kozin1969survey} surveyed some basic ideas for the stability of stochastic systems and Feng \etal\cite{feng1992stochastic} showed equivalence of different notions in mean square stabilities. In \cite{ji1991stability}, Ji \etal studied the stability for the discrete-time jump linear systems in the mean square sense. One example of such stochastic jump systems is Markov jump linear systems, where the switching probability $\pi(k)$ is governed by Markovian process. These Markov jump linear systems are widely adopted due to the usefulness of formulating randomness caused by communication delays %(\cite{xiao2000control}) 
 or packet losses %(\cite{xiong2007stabilization}) 
 in networked systems, abrupt environmental disturbances or changes in subsystems%(\cite{mariton1990jump})
 , systems with parametric uncertainties%(\cite{zhang2008analysis})
 , etc. Some recent literatures for the stability analysis of Markov jump linear systems can be found in \cite{de2006robust,shi2009output,zhang2008analysis,bolzern2010markov}.
%However, few literatures, investigated the stability conditions for the general non-Markovian stochastic jump linear systems can be found(\cite{mariton1990jump,feng1992stochastic}). In most cases, the mean square stability conditions are obtained from the Lyapunov's method in the state level, where finding Lyapunov's function, satisfying certain conditions is troublesome. 

In this note, we provide the mean square stability conditions for general stochastic jump linear systems, but not necessarily for Markovian jump. Hence, the switching probability $\pi(k)$ forms any random vector. The Wasserstein distance, which defines a metric on the manifold of probability density functions(PDFs) provides a distance between a time-varying state PDF and a reference PDF. As a consequence, the stability in terms of the distributional sense can be obtained using the Wasserstein metric. Moreover, we show that the convergence in the Wasserstein distance with Dirac as a reference PDF implies the stability in the mean square sense. Compared to the previous literatures investigated on the mean square stability of stochastic jump linear systems, the major contributions of this note can be listed as follows:
1) Using the Wasserstein distance known as an optimal transport, a \textit{unified framework} to prove the mean square stability of any stochastic jump linear systems is presented. We show that any arbitrary initial state distributions can be represented by mixture of Gaussian(MoG), then the convergence of this MoG implies the mean square stability in the Wasserstein framework.
2) Some stability conditions for stochastic jump linear systems in the previous literatures such as i) Independent and identically distributed(i.i.d.) jump, ii) Time-homogeneous Markov jump, can be fully recovered by the proposed methods. This work will show how the mean square stability conditions are established in the PDF level using an optimal transport.

\textbf{Notation:} Most notations are standard. $\mathbb{R}^{n}$ denotes the $n$-dimensional Euclidean space. The notations $\tr\left(\cdot\right)$, $\vect(\cdot)$, and $\diag(\cdot)$ represent trace, vectorization, and block diagonalization operators, respectively. Abbreviation m.s. stands for the \textit{asymptotic} convergence in the mean square sense. The notation $X \sim \varsigma\left(x\right)$ denotes that the random vector $X$ has probability density function (PDF) $\varsigma\left(x\right)$. The symbol $\mathcal{N}\left(\mu,\Sigma\right)$ is used to denote the PDF of a Gaussian random vector with mean $\mu$ and covariance $\Sigma$. In addition, the symbol $\rho(\cdot)$ and $\lambda_i(\cdot)$ represent the spectral radius and $i^{th}$ eigenvalue of the square matrix, respectively.

%The symbol $\text{supp}\left(\cdot\right)$ denotes the support of a function. %We do not make any notational distinction between the scalar value zero, a vector and a matrix with all zero entries, i.e., all of the previous three are denoted as ``0", and the meaning should be obvious from the context.

\section{Some New Results on Wasserstein Distance and PDF Evolution in Stochastic Jump Linear Systems}
\begin{definition} (\textbf{Wasserstein distance})
Consider the vectors $x_{1} \in \mathcal{X}_{1} \subseteq \mathbb{R}^{n}$, and $x_{2} \in \mathcal{X}_{2} \subseteq \mathbb{R}^{n}$, such that $x_{1} \sim \varsigma_{1}$ and $x_{2} \sim \varsigma_{2}$. Let $\mathcal{P}_{2}(\varsigma_{1},\varsigma_{2})$ denote the collection of all probability measures $\varsigma$ supported on the product space $\mathcal{X}_{1}\times\mathcal{X}_{2} \subseteq \mathbb{R}^{2n}$, having finite second moment, with first marginal $\varsigma_{1}$ and second marginal $\varsigma_{2}$. Then the $L_{2}$ \emph{Wasserstein distance of order} 2, denoted as $_{2}W_{2}$, between two $n$-variate PDFs $\varsigma_{1}$ and $\varsigma_{2}$, is defined as
\small
\begin{align}\label{Wassdefn}
&_{2}W_{2}(\varsigma_{1},\varsigma_{2}) \triangleq\nonumber\\
&\left(\displaystyle\inf_{\varsigma\in\mathcal{P}_{2}(\varsigma_{1},\varsigma_{2})}\displaystyle\int_{\mathcal{X}_{1}\times\mathcal{X}_{2}} \parallel x_{1}-x_{2}\parallel_{\ell_{2}\left(\mathbb{R}^{n}\right)}^{2} \: \varsigma\left(x_{1},x_{2}\right) \: dx_{1} dx_{2} \right)
^{\frac{{1}}{2}}.
\end{align}
\normalsize
\end{definition}

\begin{remark}
Intuitively, Wasserstein distance equals the \emph{least amount of work} needed to morph one distributional shape to the other \cite{villani2003topics}. From this point on, we denote $_{2}W_{2}$ as $W$, for notational ease. One can prove (p. 208, \cite{villani2003topics}) that $W$ defines a metric on the manifold of PDFs.%, meaning (i) $W\left(\varsigma_{1}, \varsigma_{2}\right) \geq 0$ with equality iff $\varsigma_{1} = \varsigma_{2}$, (ii) $W\left(\varsigma_{1}, \varsigma_{2}\right) = W\left(\varsigma_{2}, \varsigma_{1}\right)$, and (iii) $W\left(\varsigma_{1}, \varsigma_{3}\right) \leq W\left(\varsigma_{1}, \varsigma_{2}\right) + W\left(\varsigma_{2}, \varsigma_{3}\right)$, for any PDF $\varsigma_{3}$, with $\text{supp}\left(\varsigma_{3}\right) = \mathcal{X}_{3} \subseteq \mathbb{R}^{n}$.
\label{WassRemarkFirst}
\end{remark}

Next, we connect the distributional convergence in Wasserstein metric, with the convergence in m.s. sense. For this purpose, we consider Dirac delta as a generalized PDF, formally stated below.

\begin{definition}\label{DiracPDFdefn} (\textbf{Dirac delta as generalized PDF})
In this note, we consider the Dirac delta function $\delta\left(x\right)$, where $x \in \mathbb{R}^{n}$, defined as
\begin{eqnarray*}
\delta\left(x\right) = \begin{cases}
+\infty, & x = 0,\\
0, &  x \neq 0,
\end{cases}
\end{eqnarray*}
as a generalized PDF, since $\delta\left(x\right)$ is nonnegative for all $x \in \mathbb{R}^{n}$, and $\displaystyle\int_{\mathbb{R}^{n}} \delta\left(x\right) = 1$. This generalized PDF has measure zero support, and the corresponding cumulative distribution function (CDF) is the Heaviside (unit step) function. In particular, we note that $\delta\left(x\right) = \displaystyle\lim_{S \rightarrow 0} \mathcal{N}\left(0, S\right)$ (see e.g., p. 160-161,\cite{hassani2013mathematical}).
\end{definition}

\begin{definition}\label{MSconvergencedefn} (\textbf{Mean square convergence})
A sequence of random vectors $\{X_{j}\}_{j=1}^{\infty}$ with $X_{j} \in \mathbb{R}^{n}$, is said to converge to a random vector $X \in \mathbb{R}^{n}$ in \emph{mean-square sense}, if
\begin{eqnarray*}
\displaystyle\lim_{j\rightarrow\infty}\mathbb{E}\left[\parallel X_{j} - X \parallel_{\ell_{2}\left(\mathbb{R}^{n}\right)}^{2}\right] = 0. \:\text{In short, we write $X_{j} \overset{\text{m.s.}}\longrightarrow X$}.
\end{eqnarray*}
\end{definition}

\begin{proposition}\label{m.s.stable} (\textbf{Convergence to $\delta\left(x\right)$ in $W$ metric})
If we fix Dirac delta as the reference PDF, then distributional convergence in Wasserstein metric is \emph{necessary and sufficient} for convergence in m.s. sense.\label{WConvMeanSqConvProposition}
\end{proposition}
\begin{proof}
Consider a sequence of $n$-dimensional joint PDFs $\{\varsigma_{j}\left(x\right)\}_{j=1}^{\infty}$, that converges to $\delta\left(x\right)$ in distribution, i.e., $\displaystyle\lim_{j\rightarrow\infty} W\left(\varsigma_{j}(x), \delta(x)\right) = 0$. From (\ref{Wassdefn}), we have
{\small
\begin{align}
&\displaystyle \underset{\varsigma\in\mathcal{P}_{2}(\varsigma_{j}(x),\delta(x))}{W^{2}\left(\varsigma_{j}(x), \delta(x)\right)=\inf\mathbb{E}\left[\parallel X_{j} - 0 \parallel_{\ell_{2}\left(\mathbb{R}^{n}\right)}^{2}\right]} = \mathbb{E}\left[\parallel X_{j} \parallel_{\ell_{2}\left(\mathbb{R}^{n}\right)}^{2}\right],\label{m.s.stab}
\end{align}}
where the random vector $X_{j} \sim \varsigma_{j}\left(x\right)$. The last equality follows from the fact that $\mathcal{P}_{2}(\varsigma_{j}(x),\delta(x))$ $ = \{\varsigma_{j}(x)\}$ $\forall \: j$, thus obviating the infimum. From \eqref{m.s.stab}, $\displaystyle\lim_{j\rightarrow\infty} W\left(\varsigma_{j}(x), \delta(x)\right) = 0 \Rightarrow \displaystyle\lim_{j\rightarrow\infty} \mathbb{E}\left[\parallel X_{j} \parallel_{\ell_{2}}^{2}\right] $ $= 0$, establishing distributional convergence to $\delta(x) \Rightarrow$ m.s. convergence. Conversely, m.s. convergence $\Rightarrow$ distributional convergence, is well-known \cite{grimmett2001probability} and unlike the other direction, holds for arbitrary reference measure.
\end{proof}

The next result quantifies the Wasserstein distance between a Gaussian and Dirac PDF, in terms of the parameters of the Gaussian PDF.

\begin{proposition}\label{prop:Gauss_to_Dirac}($W$ \textbf{between Gaussian and Dirac PDF})
The Wasserstein distance $W$, between two joint PDFs $\mathcal{N}\left(\mu,\Sigma\right)$ and $\delta\left(x\right)$, both supported on $\mathbb{R}^{n}$, is given by
\begin{eqnarray}\label{staticW}
W\left(\mathcal{N}\left(\mu,\Sigma\right), \delta\left(x\right)\right) = \sqrt{\parallel \mu \parallel_{\ell_{2}\left(\mathbb{R}^{n}\right)}^{2} + \: \text{tr}\left(\Sigma\right)}.
\label{GaussianDiracW}
\end{eqnarray}
\end{proposition}
\begin{proof}
Following definition \ref{DiracPDFdefn}, we write
\small{
\begin{align*}
&W^{2}\left(\mathcal{N}\left(\mu,\Sigma\right), \delta\left(x\right)\right) = \displaystyle\lim_{S\rightarrow0} W^{2}\left(\mathcal{N}\left(\mu,\Sigma\right), \mathcal{N}\left(0,S\right)\right)\\
&= \displaystyle\lim_{S\rightarrow0} \left(\parallel \mu  - 0 \parallel_{\ell_{2}\left(\mathbb{R}^{n}\right)}^{2} + \: \text{tr}\left(\Sigma + S - 2\left(\sqrt{\Sigma}S\sqrt{\Sigma}\right)^{1/2}\right)\right)\\
%&= \parallel \mu \parallel_{\ell_{2}\left(\mathbb{R}^{n}\right)}^{2} + \: \text{tr}\left(\Sigma\right) + \displaystyle\lim_{S\rightarrow0} \text{tr}\left(S - 2\left(\sqrt{\Sigma}S\sqrt{\Sigma}\right)^{1/2}\right)\\
&= \parallel \mu \parallel_{\ell_{2}\left(\mathbb{R}^{n}\right)}^{2} + \: \text{tr}\left(\Sigma\right),
\end{align*}}
\normalsize
where in the second step, we used the closed-form solution \cite{givens1984class} for Wasserstein distance between two Gaussian PDFs. Hence the result.
\end{proof}

%==========================================================
\begin{lemma}\label{lemma:2.1}
Given $m$ absolutely continuous random vectors $X_{1}, \hdots, X_{m}$, with respective CDF $F_{j}\left(x\right)$, and PDF $\varsigma_{j}\left(x\right)$, where $j = 1, 2,\hdots, m$, and $x\in\mathbb{R}^{n}$, let $X \triangleq X_{j}$ with probability $\alpha_{j} \in [0,1]$, $\displaystyle\sum_{i=1}^{m} \alpha_{j} = 1$. Then, the CDF and PDF of the $n$-dimensional random vector $X$ are given by
{\small
\begin{equation}
F\left(x\right) = \displaystyle\sum_{j=1}^{m} \alpha_{j} F_{j}\left(x\right),\quad\varsigma\left(x\right) = \displaystyle\sum_{j=1}^{m} \alpha_{j} \varsigma_{j}\left(x\right)\label{eqn:lemma3.1}.
\end{equation}}
\label{RandomVarProposition}
\end{lemma}
\begin{proof}
{\small$F\left(x\right) \triangleq \mathbb{P}\left(X \leq x\right) = \displaystyle\sum_{j=1}^{m} \mathbb{P}\left(X=X_{j}\right) \mathbb{P}\left(X_{j} \leq x\right) = \displaystyle\sum_{j=1}^{m} \alpha_{j} F_{j}\left(x\right)$}
where we have used the law of total probability. Since each $X_{j}$ and hence $X$, is absolutely continuous, we have $\varsigma\left(x\right) =  \sum_{j=1}^{m} \alpha_{j} \varsigma_{j}\left(x\right)$.
\end{proof}
A consequence of Lemma \ref{lemma:2.1} is that the joint state PDF of stochastic jump systems is of mixture type, namely a convex sum of component PDFs. In particular, the following proposition provides a closed form formula of the joint state PDF evolution for stochastic jump \emph{linear} systems, under the assumption that the initial PDF is an MoG. This assumption on initial joint PDF is not too restrictive, since any arbitrary initial PDF can be approximated, in weak distributional sense, by a finite MoG \cite{bacharoglou2010approximation}.

\begin{proposition} (\textbf{Joint state PDF of stochastic jump linear systems at time $k$})\label{SJLSpdfAtAnyTimeWithInitialMOG}
Consider a discrete-time stochastic jump linear system with the initial joint state PDF $\varsigma_{0}$ being an MoG  with $m_{0}$ component Gaussians, i.e., $\varsigma_{0} =  \sum_{j_{0}=1}^{m_{0}} \alpha_{j_{0}} \:\mathcal{N}\left(\mu_{j_{0}},\Sigma_{j_{0}}\right)$, $ \sum_{j_{0}=1}^{m_{0}} \alpha_{j_{0}} = 1$. Then, the joint state PDF at time $k$, denoted by $\varsigma\left(k\right)$, under stochastic jumps with switching probability $\pi(k)$, is given as
{\small
\begin{align}
&\varsigma\left(k\right) = \displaystyle\sum_{j_{k}=1}^{m}\displaystyle\sum_{j_{k-1}=1}^{m} \hdots \displaystyle\sum_{j_{1}=1}^{m} \displaystyle\sum_{j_{0}=1}^{m_{0}} \left(\prod_{r=1}^{k} \pi_{j_{r}}\left(r\right)\right)  \, \nonumber\\
&\qquad\qquad\qquad\qquad\qquad\alpha_{j_{0}}\mathcal{N}\left(A_{j_k}^*\mu_{j_{0}}, A_{j_k}^*\Sigma_{j_{0}}A_{j_k}^{*{\top}}\right),
\label{dtSJLSstatePDF}
\end{align}}
where {\small$\displaystyle A_{j_{k}}^*\triangleq \prod_{r=k}^{1}A_{j_r}=A_{j_k}A_{j_{k-1}}\hdots A_{j_{2}}A_{j_{1}}$}.
\label{UncPropThmdtSJLS}
\end{proposition}
\begin{proof}
Starting from $\varsigma_{0}$ at $k=0$, the modal PDF at time $k=1$, is given by
{\small
\begin{align}
\varsigma_{j}(1) &= \displaystyle\sum_{j_{0}=1}^{m_{0}}\alpha_{j_{0}} \:\mathcal{N}\left(A_{j}\mu_{j_{0}}, A_{j}\Sigma_{j_{0}}A_{j}^{\top}\right), \; j=1,\cdots ,m, \nonumber%\label{ModalPDFatTime1}
\end{align}}
which follows from the fact that linear transformation of an MoG is an equal component MoG with linearly transformed component means and congruently transformed component covariances (see Theorem 6 and Corollary 7 in \cite{ali2008convergence}). From Lemma \ref{lemma:2.1}, it follows that the state PDF at $k=1$, is
{\small
\begin{eqnarray}
\varsigma(1) = \displaystyle\sum_{j_{1}=1}^{m}\displaystyle\sum_{j_{0}=1}^{m_{0}} \pi_{j_{1}}(1) \alpha_{j_{0}} \:\mathcal{N}\left(A_{j_{1}}\mu_{j_{0}}, A_{j_{1}}\Sigma_{j_{0}}A_{j_{1}}^{\top}\right),
\label{dtSJSPDFatTime1}
\end{eqnarray}}
where $\pi_{j_{1}}(1)$ is the occupation probability for mode $j_{1}$ at time $k=1$. Notice that (\ref{dtSJSPDFatTime1}) is an MoG with $m m_{0}$ component Gaussians. Proceeding likewise from this $\varsigma(1)$, we obtain
{\small
\begin{align}
&\varsigma_{j}(2) = \displaystyle\sum_{j_{1}=1}^{m}\displaystyle\sum_{j_{0}=1}^{m_{0}}\pi_{j_{1}}(1) \alpha_{j_{0}} \:\mathcal{N}\big((A_{j}A_{j_{1}})\mu_{j_{0}},\nonumber\\
&\qquad\qquad\qquad\qquad (A_{j}A_{j_{1}})\Sigma_{j_{0}}(A_{j}A_{j_{1}})^{\top}\big), \quad j=1,\hdots,m,\nonumber\\
&\varsigma(2)=\displaystyle\sum_{j_{2}=1}^{m}\displaystyle\sum_{j_{1}=1}^{m}\displaystyle\sum_{j_{0}=1}^{m_{0}}\pi_{j_{2}}(2)\pi_{j_{1}}(1) \alpha_{j_{0}} \:\mathcal{N}\big((A_{j_{2}}A_{j_{1}})\mu_{j_{0}}, \nonumber\\
&\qquad\qquad\qquad\qquad\qquad(A_{j_{2}}A_{j_{1}})\Sigma_{j_{0}}(A_{j_{2}}A_{j_{1}})^{\top}\big).
\label{ModalStatePDFatTime2}
\end{align}}
Continuing with this recursion till time $k$, we arrive at (\ref{dtSJLSstatePDF}), which is an MoG with $m^{k}m_{0}$ components. %We comment that the expression simplifies for $m_{0} = 1$, i.e. when the initial PDF is Gaussian.
\end{proof}
%Proof is omitted due to the limit in space and it can be found in Theorem 1 of \cite{lee2013robustness}.

The next Lemma computes the mean and covariance of \emph{any mixture PDF}, in terms of the means and covariances of its component PDFs.

\begin{lemma} \label{lemma:4.2}(\textbf{Mean and covariance of a mixture PDF})%, Lemma 1 of \cite{lee2013robustness}
Consider any $q$-component mixture PDF $\varsigma(x) = \displaystyle\sum_{j=1}^{q} \beta_{j} \varsigma_{j}(x)$, with $\displaystyle\sum_{j=1}^{q} \beta_{j} = 1$, that has component mean-covariance pairs $\left(\mu_{j},\Sigma_{j}\right)$, $j=1,\hdots,q$. Then, the mean-covariance pair $\left(\widehat{\mu},\widehat{\Sigma}\right)$ for the mixture PDF $\varsigma(x)$, is given by
\small{
\begin{align}
&\widehat{\mu} = \sum_{j=1}^{q}\beta_{j} \mu_{j},\quad
\widehat{\Sigma} =\sum_{j=1}^{q} \beta_{j}\left(\Sigma_{j} + \left(\mu_{j}-\widehat{\mu}\right)\left(\mu_{j}-\widehat{\mu}\right)^{\top}\right).
\label{MeanCovHat}
\end{align}}
\label{MeanCovMixPDF}
\end{lemma}
\begin{proof}
By definition, mean vector of the mixture PDF is
\small
\begin{align*}
\widehat{\mu} &\triangleq \displaystyle\int_{\mathbb{R}^{n}} x \varsigma(x) dx
% = \displaystyle\int_{\mathbb{R}^{n}} x \displaystyle\sum_{j=1}^{q} \beta_{j} \varsigma_{j}(x) dx
 = \displaystyle\sum_{j=1}^{q} \beta_{j} \displaystyle\int_{\mathbb{R}^{n}} x \varsigma_{j}(x) dx
 = \displaystyle\sum_{j=1}^{q}\beta_{j} \mu_{j}.
\end{align*}\normalsize
Next, covariance matrix of the mixture PDF is
\small
\begin{align*}
\widehat{\Sigma} &\triangleq \mathbb{E}\left[\left(x-\widehat{\mu}\right) \left(x-\widehat{\mu}\right)^{\top}\right] = \mathbb{E}\left[x x^{\top}\right] - \widehat{\mu}\widehat{\mu}^{\top}\\
%&= \displaystyle\int_{\mathbb{R}^{n}} x x^{\top} \displaystyle\sum_{j=1}^{q} \beta_{j} \varsigma_{j}(x) dx - \widehat{\mu}\widehat{\mu}^{\top}\\
&= \displaystyle\sum_{j=1}^{q} \beta_{j} \displaystyle\int_{\mathbb{R}^{n}} \left(x - \widehat{\mu} + \widehat{\mu}\right) \left(x - \widehat{\mu} + \widehat{\mu}\right)^{\top} \varsigma_{j}\left(x\right) dx - \widehat{\mu}\widehat{\mu}^{\top}\\
&=\displaystyle\sum_{j=1}^{q}\beta_{j}\left(\Sigma_{j}+\left(\mu_{j}-\widehat{\mu}\right)\left(\mu_{j}-\widehat{\mu}\right)^{\top}\right).
\end{align*}\normalsize
\end{proof}
%Although we compute mean and covariance of a mixture PDF by Lemma \ref{lemma:4.2}, this does not imply that a synthetic Gaussian PDF with $\widehat{\mu}$ and $\widehat{\Sigma}$ can replace a mixture PDF.

From Proposition \ref{SJLSpdfAtAnyTimeWithInitialMOG}, starting with an MoG, the joint PDF for a stochastic jump linear system at any time, is another MoG, which may have higher moments other than mean and covariance. %So if we want to analyze convergence of the joint PDFs generated by a stochastic jump linear system, Lemma \ref{lemma:4.2} may seem to be of limited use. 
However, Theorem \ref{theorem:2.3} stated below, somewhat counter-intuitively shows that given an MoG joint PDF, one can construct a ``synthetic" Gaussian PDF using Lemma \ref{lemma:4.2}, such that the given MoG and the synthetic Gaussian are equidistant from Dirac PDF, as measured in the Wasserstein metric. Hence, in order to assess m.s. convergence for stochastic jump linear systems, we can use the distance $W\left(\mathcal{N}(\widehat{\mu}(k),\widehat{\Sigma}(k)),\delta\right)$ as a ``proxy" for the distance $W\left(\varsigma\left(k\right), \delta\right)$. As we will see in Section III, this will facilitate our stability analysis since computing $W\left(\mathcal{N}(\widehat{\mu}(k),\widehat{\Sigma}(k)),\delta\right)$ is much easier (Proposition \ref{prop:Gauss_to_Dirac}) than computing $W\left(\varsigma\left(k\right), \delta\right)$, as the number of Gaussian components in $\varsigma\left(k\right)$ increases exponentially with time (Proposition \ref{SJLSpdfAtAnyTimeWithInitialMOG}).

\begin{theorem}\label{theorem:2.3} (\textbf{MoG state PDF and a synthetic Gaussian are equidistant from Dirac})
Starting from an initial MoG joint PDF {\small$\varsigma_{0}\left(k\right) = \displaystyle\sum_{j_{0}=1}^{m_{0}} \alpha_{j_{0}} \mathcal{N}\left(\mu_{j_{0}}, \Sigma_{j_{0}}\right)$}, let $\varsigma(k)$ be the joint state PDF at time $k$, for stochastic jump linear systems with arbitrary switching probability $\pi(k)$. Further, let the mean and covariance for $\varsigma\left(k\right)$, be
denoted as $\widehat{\mu}(k)$ and $\widehat{\Sigma}(k)$, respectively. Let {\small$W(k)\triangleq W\left(\varsigma(k),\delta(x)\right)$}, and {\small$\widehat{W}(k)\triangleq W\left(\mathcal{N}\left(\widehat{\mu}(k),\widehat{\Sigma}(k)\right),\delta(x)\right)$}. Then
\begin{equation}
W^2(k) = \widehat{W}^2(k) = \vect(I_n)^{\top}\Gamma(k)\vect(\widehat{\mu}(0)\widehat{\mu}(0)^{\top}+\widehat{\Sigma}(0)),\label{eqn:W^2(k)}
\end{equation}
where $I_n$ denotes the $n\times n$ identity matrix. Further, {\small$
\widehat{\mu}(0) = \displaystyle\sum_{j_{0}=1}^{m_{0}} \alpha_{j_{0}} \mu_{j_{0}}$, $
\widehat{\Sigma}(0) = \displaystyle\sum_{j=1}^{m_{0}}\alpha_{j_{0}}\left(\Sigma_{j_{0}}+\right.$ $\left. \left(\mu_{j_{0}}-\widehat{\mu}(0)\right)\left(\mu_{j_{0}}-\widehat{\mu}(0)\right)^{\top}\right)$}
are the mean and covariance of $\varsigma_{0}$, respectively. The matrix $\Gamma(k)$ is defined as \small{$\displaystyle \Gamma(k)\triangleq\prod_{i=k}^{1}\left(\sum_{j=1}^{m}\pi_j(i)\left(A_j\otimes A_j\right)\right)$}\normalsize, which is the product of matrices in reverse order w.r.t. time.

The proof is given in Appendix and more details about the practicality of Theorem \ref{theorem:2.3} can be found in \cite{lee2013robustness}.
\end{theorem}
%%%%%%%%%%%%%%%%%%%%%%%%%%%%%%%%%%%%%%%%%%%%%%%%%%%%%%%%%%%%%%%%%%%%%%%%%%%%%%%%%%%%%%%%%%%%%%%%%%%%%%%%%%%%%%%%%%%%%%%%%%

\section{Main Results}
\subsection{Mean square stability for i.i.d. jump linear systems}
Suppose that $\{\sigma_k\}$ is generated by an i.i.d. process with probability distribution $\{\pi_1, \pi_2,\cdots,\pi_m\}$ over the set $\{1, 2, \hdots, m\}$. In Corollary 2.7 of \cite{fang2002stochastic}, the necessary and sufficient condition for m.s. stability of an i.i.d. jump linear system is given by that the matrix
\begin{eqnarray}
A \triangleq \sum_{j=1}^{m}\pi_j\left(A_j \otimes A_j\right) & = \pi_1 (A_1\otimes A_1) + \pi_2 (A_2\otimes A_2) +\nonumber\\
& \cdots + \pi_m (A_m\otimes A_m)\label{eqn:iid_Amatrix}
\end{eqnarray}
is Schur stable. We next recover this result from the Wasserstein distance perspective.
\begin{theorem}
Consider an i.i.d. jump linear system, where $\pi\left(k\right)$ is a stationary probability vector $\{\pi_1, \pi_2,\cdots,\pi_m\}$ for all $k$. The i.i.d. jump linear system is m.s. stable iff the matrix $A \triangleq \displaystyle\sum_{j=1}^{m}\pi_j\left(A_j \otimes A_j\right)$ is Schur stable, i.e. $\rho\left(A\right) < 1$.
\end{theorem}
\begin{proof}
Since the jump stochastic process is i.i.d., the underlying probability vector $\pi\left(k\right)$ that generates the switching sequence $\{\sigma_{k}\}$, is a time-invariant probability vector $\{\pi_1,\pi_2,\cdots,\pi_m\}$. As a consequence, \eqref{eqn:W^2(k)} can be simplified as
$W^2(k) = \vect(I_n)^{\top}(A^k)$ $\vect(\widehat{\mu}(0)\widehat{\mu}(0)^{\top}+\widehat{\Sigma}(0))$, where $A = \left(\sum_{j=1}^{m}\pi_j\left(A_j\otimes A_j\right)\right)$.
However, it is well known that $\displaystyle \lim_{k\rightarrow\infty}A^k = 0$ iff $\rho(A) < 1$. Therefore, $\displaystyle\lim_{k\rightarrow\infty}W^2\rightarrow0 \Leftrightarrow \rho(A) < 1$. In addition, proposition \ref{m.s.stable} tells us that $\displaystyle\lim_{k\rightarrow\infty}W\rightarrow0 \Leftrightarrow \text{m.s. stability}$. Combining these two, we arrive at $\rho(A) < 1 \Leftrightarrow \displaystyle\lim_{k\rightarrow\infty}W\rightarrow0 \Leftrightarrow$ m.s. stability for i.i.d. jump linear system.
\end{proof}

\subsection{Mean square stability for Markov jump linear systems}
Suppose that $\{\sigma_k\}$ is generated by a time-homogeneous Markov chain with probability distribution $\pi\left(k\right) = \{\pi_{1}\left(k\right), \pi_{2}\left(k\right), \hdots, \pi_{m}\left(k\right)\}$, satisfying $\pi\left(k+1\right) = \pi\left(k\right) P$, where $P=(p_{ij})$ is the transition probability matrix of size $m \times m$. It has been shown in \cite{costa1993stability} (see Theorem 1 and 2 therein) that the condition $\rho\left(\diag(A_j\otimes A_j)\left(P^{\top}\otimes I\right)\right)$ $< 1$ is necessary and sufficient for the m.s. stability of the Markov jump linear systems. Also, it turns out \cite{costa1993stability} that the Markov chain for the jump process admits a stationary probability distribution $\pi^{*}$ satisfying $\pi^{*} = \pi^{*}P$, needs to be enforced for the spectral radius condition to imply m.s. stability.

Now we recover this m.s. stability condition in the Wasserstein framework. The following matrix properties are needed for this purpose.
\begin{lemma}\label{lemma:matrix prop.}
For any real matrices $X$, $Y$ $\in\mathbb{R}^{n\times n}$, let the matrices $M\in\mathbb{R}^{n\times n}$ and  $N\in\mathbb{R}^{nm\times nm}$ be of the form
\small
\begin{align*}
M = X_1 + X_2 + \cdots X_m,\:
N=\begin{bmatrix}
X_1 & X_1 & \cdots & X_1\\
X_2 & X_2 & \cdots & X_2\\
\vdots & \vdots & \ddots & \vdots\\
X_m & X_m & \cdots & X_m
\end{bmatrix},
\end{align*}
\normalsize
%where index $i$, $j$ in $X_i$ denotes that $X_i = X_p^{(q)}$ only if $i=p$, $j=q$.
Then, following properties hold true.
\begin{description}
\item[(a)] $(XY)\otimes I = (X\otimes I)(Y\otimes I)$
\item[(b)] $\lambda_i(M) = \lambda_i\left(N\right),\quad i=1,2,\hdots,n$
\item[(c)] $\displaystyle \lambda_i(\prod_{j=1}^{k}M_j)=\lambda_i(\prod_{j=1}^{k}N_j), \quad i=1,2,\hdots,n$
\end{description}
\end{lemma}
\begin{proof}[Proof of (a)] From the mixed-product property of the Kronecker product, it is known that $(A\otimes B)(C\otimes D) = AC\otimes BD$.
Similarly, we have $(X\otimes I)(Y\otimes I) = XY\otimes I$.
\end{proof}
\begin{proof}[Proof of (b)] Let $\lambda$ and $v$ be the eigenvalues and corresponding eigenvectors of block matrix $N$, respectively. Then, $N$ satisfies
\begin{align}
Nv = \lambda v.\label{Mv=lv}
\end{align}
Suppose that the eigenvector $v$ has the form $\displaystyle v = [v_1^{\top}, v_2^{\top}, \hdots, v_m^{\top}]^{\top}$ with $v_i \in \mathbb{R}^{n}$, $i=1,2,\hdots,m$.
Then, we have
\small
\begin{align}
&Nv = \begin{bmatrix}
X_1 & X_1 & \cdots & X_1\\
X_2 & X_2 & \cdots & X_2\\
\vdots & \vdots & \ddots & \vdots\\
X_m & X_m & \cdots & X_m\\
\end{bmatrix}
\begin{bmatrix}
v_1\\
v_2\\
\vdots\\
v_m
\end{bmatrix}
=\begin{bmatrix}
\lambda v_1\\
\lambda v_2\\
\vdots\\
\lambda v_m
\end{bmatrix}\nonumber \\
\nonumber \\
&\Rightarrow
\begin{matrix}
X_i(v_1+v_2+\cdots+v_m) = \lambda v_i,\quad i=1,2,\hdots,m.
\end{matrix}\label{Xv=lv}
\end{align}
\normalsize
By adding all $m$ equations in \eqref{Xv=lv} we obtain
\begin{align}
Mw = \lambda w\label{Nw=lw},
\end{align}
where $w \triangleq (v_1+v_2+\cdots + v_m)\in \mathbb{R}^n$ is the eigenvector of the matrix $M\triangleq (X_1+X_2+\cdots+X_m)\in \mathbb{R}^{n\times n}$.
From \eqref{Mv=lv} and \eqref{Nw=lw}, we know that $M$ and $N$ share same eigenvalues.
Since rank$(N)\leq n$, all the remaining $(m-1)n$ eigenvalues of $N$ should always be zero, i.e., $\lambda_{i}(M) = \lambda_{i}(N)$ for $i=1,2,\cdots,n$ and $\lambda_{i}(N) = 0$ for $i=n+1,n+2,\cdots,mn$.
\end{proof}
\begin{proof}[Proof of (c)] Firstly, we compute product of two matrices as follows.
\small
\begin{align*}
M_1M_2 &= \underbrace{(X_{1}+X_{2}+\cdots +X_{m})}_{M_1}\underbrace{(Y_{1}+Y_{2}+\cdots +Y_{m})}_{M_2}\\
&= (X_{1}M_2+X_{2}M_2+\cdots +X_{m}M_2),\\
%N_1N_2&=\underbrace{\begin{bmatrix}
%X_{1} & X_{1} & \cdots & X_{1}\\
%X_{2} & X_{2} &\cdots & X_{2}\\
%\vdots & \ddots & \vdots \\
%X_{m} & X_{m} &\cdots & X_{m}\\
%\end{bmatrix}}_{N_1}
%\underbrace{\begin{bmatrix}
%X_{1} & X_{1} & \cdots & X_{1}\\
%X_{2} & X_{2} & \cdots & X_{2}\\
%\vdots & \ddots & \vdots \\
%X_{m} & X_{m} & \cdots & X_{m}\\
%\end{bmatrix}}_{N_2}
%\\
N_1N_2&=\begin{bmatrix}
X_{1}M_2 & X_{1}M_2 & \cdots & X_{1}M_2\\
X_{2}M_2 & X_{2}M_2 & \cdots & X_{2}M_2\\
\vdots & \vdots & \ddots & \vdots \\
X_{m}M_2 & X_{m}M_2 & \cdots & X_{m}M_2\\
\end{bmatrix}.
\end{align*}
\normalsize
Then, by Lemma \ref{lemma:matrix prop.}(b) we know that $\lambda_i(M_1M_2) = \lambda_i(N_1N_2)$, $i=1,2,\cdots,n$.
Similarly, we have
\small
\begin{align*}
M_1M_2M_3 &= (X_{1}M_2M_3+X_{2}M_2M_3+\cdots +X_{m}M_2M_3),\\
N_1N_2N_3 &= \begin{bmatrix}
X_{1}M_2M_3 & X_{1}M_2M_3 & \cdots & X_{1}M_2M_3\\
X_{2}M_2M_3 & X_{2}M_2M_3 & \cdots & X_{2}M_2M_3\\
\vdots & \vdots & \ddots & \vdots \\
X_{m}M_2M_3 & X_{m}M_2M_3 & \cdots & X_{m}M_2M_3\\
\end{bmatrix},
\end{align*}
\normalsize
resulting in {\small$\lambda_i(M_1M_2M_3) = \lambda_i(N_1N_2N_3)$, $i=1,2,\hdots,n$}. Proceeding likewise, we get {\small$\displaystyle \lambda_i(\prod_{j=1}^{k}M_j)=\lambda_i(\prod_{j=1}^{k}N_j)$, $i=1,2,\hdots,n$}.
\end{proof}

\begin{theorem}\label{thm:4.2}
Suppose that $\{\sigma_k\}$ is a time-homogeneous finite state Markov chain with transition probability matrix $P$, then $W\rightarrow0$ and hence \eqref{eqn:1} is m.s. stable if and only if
\begin{align*}
\rho\left(\diag(A_j\otimes A_j)(P^{\top}\otimes I)\right) < 1.
\end{align*}
\end{theorem}

\begin{proof}
Let the time-varying matrix $\widetilde{A}(k)$ be of the form:
\small
\begin{align*}
\begin{matrix}
\widetilde{A}(k) =
\begin{bmatrix}
\pi_1(k)(A_1\otimes A_1)  & \cdots & \pi_1(k)(A_1\otimes A_1)\\
\pi_2(k)(A_2\otimes A_2)  & \cdots & \pi_2(k)(A_2\otimes A_2)\\
\vdots&\ddots&\vdots\\
\pi_m(k)(A_m\otimes A_m)  & \cdots & \pi_m(k)(A_m\otimes A_m)
\end{bmatrix}
\end{matrix}.\nonumber
\end{align*}\normalsize
Then, the matrix $\widetilde{A}(k)$ has the following equivalent form:
\small
\begin{align}
&\widetilde{A}(k) = \diag(A_j\otimes A_j)\left(
\begin{bmatrix}
\pi_1(k)&\cdots&\pi_1(k)\\
\pi_2(k)&\cdots&\pi_2(k)\\
\vdots  &\ddots&\vdots\\
\pi_m(k)&\cdots&\pi_m(k)
\end{bmatrix}
\otimes I\right)\nonumber\\
&= \diag(A_j\otimes A_j)
\left(
[\pi(k)^{\top} \: \cdots \:\pi(k)^{\top}]
\otimes I\right)
\label{eqn:13}\\
&= \diag(A_j\otimes A_j)
\left(P^{\top}
[\pi(k-1)^{\top} \: \cdots \:\pi(k-1)^{\top}]
\otimes I\right)
\label{eqn:rho(k)}\\
&= \diag(A_j\otimes A_j)
\left(
P^{\top}
\otimes I\right)
\underbrace{
\left(
[\pi(k-1)^{\top} \: \cdots \:\pi(k-1)^{\top}]
\otimes I\right)}_{\triangleq Q(k-1)}.\label{eqn:rho(k)_last}
\end{align}\normalsize
From \eqref{eqn:13} to \eqref{eqn:rho(k)} we used a probability distribution update rule in Markov chain $\pi(k)=\pi(k-1)P$, where $P$ is the transition probability matrix.
Also, Lemma \ref{lemma:matrix prop.}(a) was applied to above equations from \eqref{eqn:rho(k)} to \eqref{eqn:rho(k)_last}.

Notice that from \eqref{eqn:13} we have
\small
\begin{align}
\widetilde{A}(k-1)= \diag(A_j\otimes A_j)Q(k-1).\label{eqn:rho(k-1)}
\end{align}
\normalsize
According to \eqref{eqn:rho(k)_last} and \eqref{eqn:rho(k-1)}, we can infer that $\widetilde{A}(k)$ can be expressed in terms of $\pi(0)$ as
\small
\begin{align}
\widetilde{A}(k) = \diag(A_j\otimes A_j)(P^{\top}\otimes I)^kQ(0),
\end{align}\normalsize
where $Q(0)\triangleq\left(
[\pi(0)^{\top} \: \cdots \:\pi(0)^{\top}]
\otimes I\right)$. 
In addition, we define the matrix $\widetilde{\Gamma}(k)$ by the product of matrix $\widetilde{A}$ from time $k$ to $1$ as follows.
\begin{align}
\widetilde{\Gamma}(k) \triangleq \widetilde{A}(k)\widetilde{A}(k-1)\cdots\widetilde{A}(2)\widetilde{A}(1) = \prod_{i=k}^{1}\widetilde{A}(i).\label{eqn:19}
\end{align}
Taking the limit of \eqref{eqn:19}, above equation has the form of
\begin{align}
\lim_{k\rightarrow\infty}\widetilde{\Gamma}(k)
&= \lim_{k\rightarrow\infty}\prod_{i=k}^{1}\diag(A_j\otimes A_j)(P^{\top}\otimes I)^{i}Q(0),\label{eqn:20}
\end{align}

Note that if the finite state time-homogeneous Markov chain has stationary probability distribution, which is $\pi^* = \pi^* P$, then the transition probability matrix $P$ satisfies $P^* = PP^* = P^*P$, where  $P^*\triangleq\displaystyle\lim_{k\rightarrow\infty}P^k$. With the fact that $\displaystyle \lim_{k\rightarrow\infty}(P^{\top}\otimes I)^k = \lim_{k\rightarrow\infty}\left((P^{\top})^k\otimes I^k\right) = {P^*}^{\top}\otimes I$, once $P$ reaches $P^*$, then we have infinite number of multiplications for the term $\diag(A_j\otimes A_j)\left({P^*}^{\top}\otimes I\right)Q(0)$ in \eqref{eqn:20}. Thus, \eqref{eqn:20} becomes
\small
\begin{align}
\lim_{k\rightarrow\infty}\widetilde{\Gamma}(k)
= \lim_{k\rightarrow\infty}\Big(\diag(A_j\otimes A_j)\left({P^*}^{\top}\otimes I\right)Q(0)\Big)^k\cdots\nonumber\\
\cdots\Big(\diag(A_j\otimes A_j)\left({P}^{\top}\otimes I\right)Q(0)\Big).\label{eqn:21}
\end{align}
\normalsize
Furthermore, the first term of the right hand side in \eqref{eqn:21} can be expressed as
\small{
\begin{align}
&\lim_{k\rightarrow\infty}\Big(\diag(A_j\otimes A_j)\left({P^*}^{\top}\otimes I\right)Q(0)\Big)^k\nonumber\\
&= \quad\lim_{k\rightarrow\infty}\Big(\diag(A_j\otimes A_j)\big({(P^*P)}^{\top}\otimes
 I\big)Q(0)\Big)^k\nonumber\\
&\overset{\mathclap{\text{Lemma}\ref{lemma:matrix prop.}(a)}}= \quad\lim_{k\rightarrow\infty}\Big(\diag(A_j\otimes A_j)\left({P}^{\top}\otimes I\right)\big({P^*}^{\top}\otimes
 I\big)Q(0)\Big)^k\nonumber\\
&= \quad\lim_{k\rightarrow\infty}\Big(\diag(A_j\otimes A_j)\left({P}^{\top}\otimes I\right)\Big)^k\Big(\big({P^*}^{\top}\otimes
 I\big)Q(0)\Big)^k.\label{eqn:22}
\end{align}}
\normalsize
Now, according to \eqref{eqn:22}, {\small$\displaystyle \lim_{k\rightarrow\infty}\Big(\diag(A_j\otimes A_j)$ $\left({P}^{\top}\otimes I\right)\Big)^k = 0\Leftrightarrow \displaystyle\lim_{k\rightarrow\infty}\widetilde{\Gamma}(k) = 0$} if and only if {\small$\rho\left(\diag(A_j\otimes A_j)(P^{\top}\otimes I)\right) < 1$}.

Finally, by Lemma \ref{lemma:matrix prop.}(c) we can conclude that {\small$\displaystyle
\lim_{k\rightarrow\infty}\widetilde{\Gamma}(k)=0 \Leftrightarrow \lambda_j(\lim_{k\rightarrow\infty}\widetilde{\Gamma}(k))=0$, $j=1,2,\hdots,mn^2
\overset{\text{Lemma}\ref{lemma:matrix prop.}(c)}{\Longleftrightarrow} \displaystyle \lambda_j(\lim_{k\rightarrow\infty}\Gamma(k))=0$, $j=1,2,\hdots,n^2
\Leftrightarrow \displaystyle\lim_{k\rightarrow\infty}\Gamma(k)=0 \Leftrightarrow \lim_{k\rightarrow\infty}W(k)\rightarrow0$}, where $\Gamma(k)$ is defined in \eqref{eqn:W^2(k)}.
Consequently, by proposition \ref{prop:Gauss_to_Dirac}, jump linear system \eqref{eqn:1} with Markovian jumps is m.s. stable iff {\small$\rho\left(\diag(A_j\otimes A_j)(P^{\top}\otimes I)\right) < 1$}.
\end{proof}

\subsection{Mean square stability for general stochastic jump linear systems}
In this section, we provide a general m.s. stability condition for stochastic jump linear systems, where the jump sequence $\{\sigma_k\}$ forms any arbitrary random sequences.
\begin{theorem}\label{theorem:3.3}
For any arbitrary switching sequence $\{\sigma_k\}$ with switching probability $\pi(k)$, a jump linear system \eqref{eqn:1} is m.s. stable if and only if the matrix $\Gamma(k)$ is convergent to zero matrix, where the matrix $\Gamma(k)$ is given in Theorem \ref{theorem:2.3}.
\end{theorem}
\begin{proof}
The sufficiency is obvious because from \eqref{eqn:W^2(k)}, $W^2\rightarrow0 \Rightarrow W\rightarrow0$ implies \eqref{eqn:1} is m.s. stable by Proposition \ref{m.s.stable}, if $\displaystyle \lim_{k\rightarrow\infty}\Gamma(k)=0$. 

For the proof of the necessity, suppose that $\Gamma(k)$ is not convergent to zero as $k\rightarrow \infty$. Then, $W$ never reaches zero by \eqref{eqn:W^2(k)}, which contradicts the m.s. stability.
\end{proof}

\begin{corollary}\label{corollary:3.1}
Suppose that $\{\sigma_k\}$ is an arbitrary switching sequence of the jump linear system \eqref{eqn:1} with the occupation probability $\pi(k)$, satisfying $\pi_{i_k}(k)=1$ and $\pi_{j_k}(k)=0$, $\forall i_k\neq j_k$ for all time $k$. Then, the jump linear system \eqref{eqn:1} is m.s. stable iff there is a finite time $k$ such that
\begin{align}
\parallel A_{i_{k}}A_{i_{k-1}}\cdots A_{i_{2}}A_{i_{1}}\parallel < 1,\label{eqn:corollary}
\end{align}
where $A_{i_j}\in\{A_1, A_2, \hdots, A_m\}$, $\forall j$ and $\parallel \cdot \parallel$ denotes any matrix norm.
\end{corollary}
\begin{proof}
If $\pi(k)$ obeys $\pi_{i_k}(k)=1$ and $\pi_{j_k}(k)=0$, $\forall i_k\neq j_k$ for all $k$, then the matrix $\Gamma(k)$ in \eqref{eqn:W^2(k)} becomes {\small$ \Gamma(k) = \prod_{p=k}^{1}\left(A_{i_p}\otimes A_{i_p}\right)$}. Since {\small$\parallel \prod_{j=k}^{1}A_{i_j} \parallel =$ $\left(\parallel \prod_{j=k}^{1}\left(A_{i_j}\otimes A_{i_j}\right) \parallel\right)^{\frac{1}{2}}$}, it is easily shown that {\small$\parallel \left(A_{i_{k}}\otimes A_{i_{k}}\right) \left(A_{i_{k-1}}\otimes A_{i_{k-1}}\right)\cdots \left(A_{i_{2}}\otimes A_{i_{2}}\right)\left(A_{i_{1}}\otimes A_{i_{1}}\right)\parallel < 1$ $\Leftrightarrow$ $\parallel A_{i_{k}}A_{i_{k-1}}\cdots A_{i_{2}}A_{i_{1}}\parallel < 1$}.
Therefore, $\Gamma(k)$ is a contraction mapping and hence $W\rightarrow0\Leftrightarrow$ m.s. stable, if $\parallel A_{i_{k}}A_{i_{k-1}}\cdots A_{i_{2}}A_{i_{1}}\parallel < 1$, $\forall k$.
The necessity can be proved by contradiction, similarly with the proof in Theorem \ref{theorem:3.3}. 
\end{proof}

Although in \cite{lin2009stability}, the authors addressed a global uniform asymptotic stability, the m.s. stability condition \eqref{eqn:corollary} coincides with the condition in Theorem 6 of \cite{lin2009stability}. In the case that the initial distribution is given by Dirac PDF located at arbitrary $x_0$, i.e., $\varsigma_0 = \delta(x-x_0)$, we can also recover the uniform stability because $\mathbb{E}[\cdot]$ can be obviated in \eqref{m.s.stab}.

\section{Conclusion}
This technical note investigated the m.s. stability for the discrete-time stochastic jump linear systems using an optimal transport. The Wasserstein distance, which defines a metric on the manifold between PDFs provides a unified framework to prove the m.s. stability conditions. Without assuming any structure on the underlying jump process, we presented the general m.s. stability conditions via optimal transport. Already-known stability conditions for i.i.d. or Markov jump linear systems are also recovered from the convergence of the Wasserstein metric.

% \appendix
\appendix
%\noindent\textit{Proof of Proposition \ref{proposition:2.3}}
\noindent
\begin{proof}[Proof of Theorem \ref{theorem:2.3}] From (\ref{Wassdefn}) and (\ref{eqn:lemma3.1}), we have
\small
\begin{align}
\nonumber W^{2} &= \displaystyle\int_{\mathbb{R}^{n}} \parallel x \parallel_{\ell_{2}\left(\mathbb{R}^{n}\right)}^{2} \varsigma(x) dx
= \displaystyle\int_{\mathbb{R}^{n}} \parallel x \parallel_{\ell_{2}\left(\mathbb{R}^{n}\right)}^{2} \displaystyle\sum_{j=1}^{m} \pi_{j} \varsigma_{j}(x) dx\\
&= \displaystyle\sum_{j=1}^{m} \pi_{j} \displaystyle\int_{\mathbb{R}^{n}} \parallel x \parallel_{\ell_{2}\left(\mathbb{R}^{n}\right)}^{2} \varsigma_{j}(x) dx
= \displaystyle\sum_{j=1}^{m} \pi_{j} W_{j}^{2},\label{eqn:A.1}
\end{align}
\normalsize
where $W_j \triangleq W\left(\varsigma_j(x),\delta(x)\right)$.

Also, we can compute $\widehat{W}^2\triangleq W^2(\mathcal{N}(\widehat{\mu},\widehat{\Sigma}),\delta(x))$ from the proposition \ref{prop:Gauss_to_Dirac} as follows.
\small
\begin{align}
\widehat{W}^{2} &= \parallel \widehat{\mu}\parallel_{\ell_{2}\left(\mathbb{R}^{n}\right)}^{2} + \text{tr}(\widehat{\Sigma}) \nonumber\\
&\overset{(\ref{MeanCovHat})}{=}
\widehat{\mu}^{\top}\widehat{\mu} + \:\text{tr}\left(\displaystyle\sum_{j=1}^{m}\pi_j(\Sigma_{j} + (\mu_{j}-\widehat{\mu})(\mu_{j}-\widehat{\mu})^{\top}\right).
\label{dtSJLSThmProof}
\end{align}
\normalsize
Since $\tr(\cdot)$ is linear operator and $\displaystyle\sum_{j=1}^{m}\pi_j = 1$, \eqref{dtSJLSThmProof} becomes
\small
\begin{align}
\widehat{W}^{2} &= \widehat{\mu}^{\top}\widehat{\mu}
+ \: \displaystyle\sum_{j=1}^{m}\pi_{j}\text{tr}\left(\Sigma_{j}\right) + \: \text{tr}\left(\displaystyle\sum_{j=1}^{m} \pi_{j}\mu_{j}\mu_{j}^{\top}\right) -  \nonumber\\
&\text{tr}\left(\left(\displaystyle\sum_{j=1}^{m} \pi_{j}\mu_{j}\right)\widehat{\mu}^{\top}\right) -
\text{tr}\left(\widehat{\mu}\left(\displaystyle\sum_{j=1}^{m} \pi_{j}\mu_{j}\right)^{\top}\right) + \text{tr}\left(\widehat{\mu}\widehat{\mu}^{\top}\right).\label{maineq1}
\end{align}
\normalsize

Now, we recall from (\ref{MeanCovHat}) that $\widehat{\mu} = \displaystyle\sum_{j=1}^{m} \pi_{j} \mu_{j}$, and that $\widehat{\mu}^{\top}\widehat{\mu} = \:\text{tr}\left(\widehat{\mu}^{\top}\widehat{\mu}\right) = \:\text{tr}\left(\widehat{\mu}\widehat{\mu}^{\top}\right)$.
Consequently, the first, fourth, fifth and sixth terms in (\ref{maineq1}) cancel out, resulting in
\small
\begin{align}
\widehat{W}^{2} &= \displaystyle\sum_{j=1}^{m}\pi_{j}\text{tr}\left(\Sigma_{j}\right) + \: \displaystyle\sum_{j=1}^{m} \pi_{j} \: \text{tr}\left(\mu_{j}\mu_{j}^{\top}\right) \nonumber\\
&= \displaystyle\sum_{j=1}^{m}\pi_{j} \left(\parallel \mu_{j} \parallel_{\ell_{2}\left(\mathbb{R}^{n}\right)}^{2} + \: \text{tr}\left(\Sigma_{j}\right)\right)\label{eqn:A.4} = \displaystyle\sum_{j=1}^{m}\pi_{j} W_{j}^{2}.%\label{eqn:A.5}
\end{align}
\normalsize
From \eqref{eqn:A.1} and \eqref{eqn:A.4} $W^2 = \widehat{W}^2$ for all $k$. Therefore, we have
\small
\begin{align}
W^{2}(k) &= \widehat{W}^2(k) = \sum_{j=1}^{m}\pi_{j}(k) \left(\parallel \mu_{j}(k) \parallel_{\ell_{2}\left(\mathbb{R}^{n}\right)}^{2} + \: \text{tr}\left(\Sigma_{j}(k)\right)\right)\nonumber\\
&= \tr\left(\sum_{j=1}^{m}\pi_{j}(k)\left(\mu_{j}(k)\mu_{j}(k)^{\top}+\Sigma_{j}(k)\right)\right).\label{eqn:A.6}
\end{align}
\normalsize
Here, $\mu_{j}(k)$ and $\Sigma_{j}(k)$ are mean and covariance of the components of the Gaussian mixture at time $k$, respectively, obtained from the synthetic Gaussian $\mathcal{N}(\widehat{\mu}(k-1), \widehat{\Sigma}(k-1))$ at time $k-1$, according to
\begin{eqnarray}
\mu_{j}(k) = A_{j}\widehat{\mu}(k-1)\label{eqn_mui(k)},\quad
\Sigma_{j}(k) = A_{j}\widehat{\Sigma}(k-1)A_{j}^{\top}\label{eqn_sigmai(k)}.
\end{eqnarray}
Replacing $\mu_j(k)$ and $\Sigma_j(k)$ in \eqref{eqn:A.6} with \eqref{eqn_mui(k)}, we have
\small{
\begin{align}\label{proof_stability}
W^{2}(k) &= \tr\left(\sum_{j=1}^{m}\pi_{j}(k)A_{j}\underbrace{\left(\widehat{\mu}(k-1)\widehat{\mu}(k-1)^{\top}+\widehat{\Sigma}(k-1)\right)}_{\triangleq \Phi(k-1)}A_{j}^{\top}\right)\nonumber\\
%&= \tr\left(\sum_{j=1}^{m}\pi_{j}(k)A_{j}^{\top}A_{j}\Phi(k-1)\right)\nonumber\\
&= \tr\left(\left(\sum_{j=1}^{m}\pi_{j}(k)A_{j}^{\top}A_{j}\right)^{\top}\Phi(k-1)\right).
\end{align}
}
\normalsize
Since the trace is invariant under cyclic permutation, the property $\tr(ABC) = \tr(CAB)$ was applied between first and second line of above equations.
Moreover, using the trace property $\tr(X^{\top}Y) = \vect(X)^{\top}\vect(Y)$, \eqref{proof_stability} can be expressed as
\small\begin{align}
W^{2}(k) = \vect\left(\sum_{j=1}^{m}\pi_{j}(k)A_{j}^{\top}I_{n}A_{j}\right)^{\top}
\vect\Big(\Phi(k-1)\Big)\label{eqn_before_vec(abc)},
\end{align}
\normalsize
where $I_{n}$ is $n\times n$ identity matrix.\\
By applying $\vect(ABC) = \left(C^{\top}\otimes A\right)\vect(B)$ to the first term of \eqref{eqn_before_vec(abc)}, we obtain
\small
\begin{align}\label{eqn_W^2(k)}
&W^{2}(k) = \left(\sum_{j=1}^{m}\pi_{j}(k)\left(A_{j}^{\top}\otimes A_{j}^{\top}\right) \vect\left(I_{n}\right)\right)^{\top} \vect\Big(\Phi(k-1)\Big)\nonumber\\
&= \vect\left(I_{n}\right)^{\top}\left(\sum_{j=1}^{m}\pi_{j}(k)\left(A_{j}\otimes A_{j}\right)\right) \vect\Big(\Phi(k-1)\Big).
\end{align}
\normalsize

Recalling \eqref{dtSJLSThmProof}, we have 
{\small
$W^2 = \widehat{W}^2 = \parallel \widehat{\mu}\parallel_{\ell_{2}\left(\mathbb{R}^{n}\right)}^{2} + \text{tr}(\widehat{\Sigma})
= \tr\left(\widehat{\mu}\widehat{\mu}^{\top}+\widehat{\Sigma}\right)
= \tr\left(I_n^{\top}\left(\widehat{\mu}\widehat{\mu}^{\top}+\widehat{\Sigma}\right)\right)$}.
Again, from the trace property $\tr(X^{\top}Y) = \vect(X)^{\top}\vect(Y)$, above equation with time index $k$ further becomes
{\small
\begin{align}
W^2(k)=\widehat{W}^2(k) = \vect(I_n)^{\top}\vect\Big(\Phi(k)\Big),
\end{align}}
where $\Phi(k)\triangleq \widehat{\mu}(k)\widehat{\mu}(k)^{\top} + \widehat{\Sigma}(k)$. Similarly, $W^2$ at $k-1$ becomes,%previous time can be represented by
\small
\begin{align}
W^2(k-1)=\widehat{W}^2(k-1) = \vect(I_n)^{\top}\vect\Big(\Phi(k-1)\Big).\label{eqn:35}
\end{align}
\normalsize
From the recurrence relation between \eqref{eqn_W^2(k)} and \eqref{eqn:35}, finally we conclude that
{\small
\begin{align*}
W^2(k) = \widehat{W}^2(k) &=
\vect(I_n)^{\top}\Gamma(k)\vect\Big(\Phi(0)\Big)\\
&=\vect(I_n)^{\top}\Gamma(k)\vect\left(\widehat{\mu}(0)\widehat{\mu}(0)^{\top}+\widehat{\Sigma}(0)\right),
\end{align*}}
where {\small$\displaystyle \Gamma(k) \triangleq \prod_{i=k}^{1}A(i) = A(k)A(k-1)\cdots A(2)A(1)$} and {\small$\displaystyle A(k)=\sum_{j=1}^{m}\pi_j(k)(A_j\otimes A_j)$}.
\end{proof}
% \section{}
% \label{}

%% References
%%   \citep{key}         ==>>  [#]


%\bibliographystyle{elsarticle-num}



%\begin{thebibliography}{00}
%\bibitem{do Valle Costa, Oswaldo Luiz, M. Marcelo Dutra Fragoso, and Ricardo Paulino Marques. Discrete time Markov jump linear systems. Springer, 2005.}
%
%\end{thebibliography}

\bibliographystyle{unsrt}%{unsrt}
\bibliography{W_to_MJLS}
\end{document}